\documentclass[10pt]{article}
\usepackage{latexsym, amscd, amsfonts, eucal, mathrsfs, amsmath, amssymb, amsthm, xr,  makeidx, stmaryrd}
 \usepackage{ragged2e}
\usepackage{thmtools,thm-restate}
\usepackage{hyperref}   
\usepackage{epigraph}
\usepackage{graphicx}
\usepackage[square, sort&compress]{natbib}
\usepackage[usenames]{color}
\usepackage{verbatim}
\usepackage{tensor}
\usepackage{enumitem}
\usepackage{tikz-cd}
\usetikzlibrary{decorations.pathmorphing}
\usepackage{subcaption}
\usetikzlibrary{backgrounds}
\usepackage{float}

\def\C{\mathcal{C}}
\def\D{\mathcal{D}}

\def\M{\mathcal{M}}

\def\S{\mathcal{S}}
\def\E{\mathcal{E}}
\def\G{\mathcal{G}}

\def\u{\mathfrak{u}}

\def\<{\langle}
\def\>{\rangle}

\definecolor{red}{rgb}{0.9,0,0}
\definecolor{green}{rgb}{0,0.8,0}
\definecolor{blue}{rgb}{0,0,0.8}
\definecolor{cautionred}{rgb}{1.0,0,0}

\newcommand{\bea}{\begin{eqnarray}}
\newcommand{\eea}{\end{eqnarray}}
\newcommand{\beann}{\begin{eqnarray*}}
\newcommand{\eeann}{\end{eqnarray*}}
\newcommand{\beq}{\begin{equation}}
\newcommand{\eeq}{\end{equation}}
\newcommand{\ba}{\begin{array}}
\newcommand{\ea}{\end{array}}
\newcommand{\ben}{\begin{enumerate}}
\newcommand{\een}{\end{enumerate}}

\theoremstyle{definition}

 \author{
  James Nguyen, Nicholas J. Teh, and Laura Wells}

\begin{document}

\title{Why surplus structure is not superfluous\\
\small Forthcoming in \textit{British Journal for the Philosophy of science}}
\date{\today}
\maketitle

\abstract{The idea that gauge theory has `surplus' structure poses a puzzle: in one much discussed sense, this structure is redundant; but on the other hand, it is also widely held to play an essential role in the theory. In this paper, we employ category-theoretic tools to illuminate an aspect of this puzzle. We precisify what is meant by `surplus' structure by means of functorial comparisons with equivalence classes of gauge fields, and then show that such structure is essential for any theory that represents a rich collection of physically relevant fields which are `local' in nature.}

\tableofcontents

\centering
\section{Introduction}
\justify
In recent years, there has been a movement to view scientific theories as categories, and to analyse them using the tools of category theory.\footnote{See e.g. \citep{tt2016, halvorson2012scientific, halvorson2015categories, W1, W2} and references therein.}  
One set of motivations for this approach stems from considerations about how best to represent scientific theories from a philosophical perspective.\footnote{Note that we do not claim any originality for this approach on the part of the philosophical literature; the general idea of `representing theories as categories' has been common within theoretical physics for a long time, e.g. in the applications of noncommutative geometry \citep{Connes1994} to physics.} 
For instance, if scientific theories are thought of as---in some sense---collections of models, then we might have reason to consider inter-model relationships to be an important part of the `structural' content of a theory.
As such, advocates of the categorical approach recommend it because it not only allows us to clearly and precisely formulate the `structure' of a theory (i.e. of its collection of models), but also operates with a notion of structure that is flexible enough to incorporate inter-model relationships. 
Furthermore, once such structure has been represented, the framework naturally takes it into account when formalising inter-theoretic comparisons such as `theoretical equivalence' and `one theory having more structure than another'.

An important class of theories for testing the fruitfulness of categorical methods is `gauge theories', by which we mean theories of `Yang-Mills type', such as electromagnetism. This is because such theories provide a paradigmatic scenario in which comparisons of structure between two theories (or two formulations of a theory) naturally arise.\footnote{We note that the term `gauge' is used in many different ways within the philosophical literature; for instance, some authors use `gauge' to describe any theory with representationally redundant structure. Our stricter notion of `gauge' is standard within the high-energy theory literature, and is appropriate given that `Yang-Mills type' theories are our concern in the present article.}
More specifically, it is typically held that gauge theories admit of redundancies in the following sense: locally, a theory whose models use gauge fields is `empirically equivalent' to a corresponding theory whose models use only gauge equivalence classes of such fields. 
Thus, there is an intuitive sense in which the former theory possesses `surplus' (i.e. structure over and above the gauge equivalence classes) compared to the latter theory.
Illuminating and precisifying this sense is a task for which the categorical approach seems ideally suited. 
 
The above notion of `surplus' structure has received much attention within the philosophical literature.\footnote{See for instance \citep{Earman2004, Healey2007, Redhead2003} and the discussion in \citep{PittGauge}.}
An equally intriguing aspect of such structure, albeit one that has received less attention, is the idea that despite its redundancy in the above sense, it also plays an indispensable role in the theory's formulation.
For instance, it is commonly believed that gauge fields are essential for formulating classical gauge theory as a `local' field theory (in a sense that we will make precise in Section \ref{gaugefunc}).\footnote{This view is widely held, but for an explicit and mathematically sophisticated presentation, see \citep{Schreiber:2016pxa}. For a less formal point of view, see also the discussion of `stacks' in \citep{Witten2010}.}
This leads to an apparent puzzle: How can a `redundancy' be an essential feature of a theory?

Although it is certainly not obvious that the categorical approach has anything to say about this puzzle, we will argue in this paper that it does.
In particular, we will show that using categorical tools to precisify `surplus' gives us the resources to understand why such `surplus' plays a representationally significant role in the theory.  
From a broader perspective, we take our work to be a contribution to both the philosophy of scientific theories literature and the philosophy of gauge theory literature.
It contributes to the former by showing that for a large class of scientific theories, namely field theories, it is important to extend what we call the `theories as categories' paradigm to the `theories as functors' paradigm in order to fully accommodate inter-model relationships (in this case: relationships specifying how the local systems represented by the theory can determine global systems) which are relevant for gauge theory. (We say more about why this is important in Sections 4.1.1 and 4.1.2.)
It contributes to the latter by showing that our intuitive understanding of `surplus' is ambiguous and its disambiguation rests on the precise standard of sameness that one chooses for a scientific theory; furthermore, only one of these choices will yield a satisfying resolution of the puzzle. 

The structure of the paper is as follows. 
In Section \ref{theoriescats}, we discuss the general idea of representing a scientific theory as a category. 
In Section \ref{gaugecat}, we then apply the `theories as categories' framework to an idealised version of $U(1)$ gauge theory, where we temporarily (i)  restrict our focus to contractible manifolds, and (ii) ignore the local nature of fields.
Within this framework, we offer an analysis of `surplus' that captures standard physical intuitions about what constitutes the (locally) redundant structure of a gauge theory.
For completeness, we then briefly contrast our analysis with that of \citet{W1,W2}, which investigates a different concept of `surplus' and thus does not capture the intuitions that we are after. 

In Section \ref{gaugefunc}, we remove assumptions (i) and (ii) by representing theories as functors instead of as categories. 
We then resolve our earlier puzzle by showing that from the perspective of this more realistic notion of a theory, `surplus' (in the sense picked out in Section \ref{gaugecat}) is necessary in order to accommodate a rich set of physically relevant models while simultaneously representing the `locality' of fields.
Finally, Section \ref{conc} provides concluding morals for the project of representing scientific theories categorically and points out an intriguing connection between our argument and a highly influential approach to mathematical classification that has its roots in the groundbreaking work of \citet{Grothendieck1959}. 

\centering
\section{Theories as Categories}
\label{theoriescats}
\justify
We now describe the general framework for representing a theory as a category.\footnote{In keeping with much of the literature on this subject, we assume that--in this basic case--the relevant categories are not functor categories, which would require us to antecedently define domain and co-domain categories of models. The case of functor categories will arise in what we call the `theories as functors' paradigm of Section 4.}
Section \ref{relmod} discusses an inter-model relationship that will be of interest to us: the `representational equivalence' of two models, meaning that they (accurately) represent the same possible states of affairs.
We also discuss several different kinds of categories that one might use to encode representational equivalence.
Section \ref{relthry} then reviews how functors can be used to formulate equivalence and inequivalence relations between theories.  

\centering
\subsection{Relations between models}
\label{relmod}
\justify 

Recall that a model of a theory is a structure that represents a system that is possible according to the theory.\footnote{How these structures represent physical systems is a vexing question which we will not address here. See \citep{FriggNguyen16} for an overview of options currently available in the literature.}
For instance, if we are concerned with a (classical) field theory, then a system will be encoded by some particular field configuration that satisfies the theory's equations of motion.
In the present categorical formalisation of a scientific theory $T$, such models will be the \textit{objects} of the category $\C_T$ that one is using to represent $T$.\footnote{For convenience, we use 1-categories to represent scientific theories throughout this paper, but our points are easily extended to theories that require an $n$-categorical formulation.}  
In order to fully specify $\C_T$, one also needs to define the \textit{morphisms} (or arrows) of this category, which can be taken to represent various kinds of inter-model relationships. 
In this section, we will focus on the inter-model relationships represented by \textit{isomorphisms} (i.e. invertible morphisms): Models of $\C_T$ that are related by an isomorphism are taken to represent the same possible systems, despite the fact that the models could be (strictly speaking) distinct objects.\footnote{This is motivated by the fact that isomorphic models in a category $\C$ have the `same structure' (from the perspective of $\C$), and once this structure has been chosen to capture what, and how, the model represents, then isomorphic models can (albeit not \textit{must}) be taken to be representationally equivalent.} 
Furthermore, each distinct isomorphism specifies a distinct way of transforming one model into another in a manner that preserves its representational content; such information (about the different ways relating representationally equivalent models) will play a vital role in the analysis of Section \ref{gaugefunc}. As we discuss there and in the conclusion, the isomorphisms of a theory---thought of as a category---can also contribute to the representational capacities of a theory, namely by representing the different ways in which the systems represented by the objects (models) are related to one another (e.g. as local subsystems combining into global systems).

An elementary type of category of models, then, is one in which the only inter-model relationship encoded is representational equivalence, i.e. one in which the only morphisms are isomorphisms. Such categories are called \textit{groupoids}. 
In what follows, we will have reason to consider three kinds of groupoids. First, a groupoid is a \textit{categorical set} $\S$ if and only if its only morphisms are identity morphisms.\footnote{The reason for this nomenclature is that each categorical set can be defined from a set and vice versa: if a theory has a set of models $M_T$, then we can redescribe this data as the categorical set $\S_T$ whose objects correspond to elements of $M_T$ and whose morphisms are all and only identities (and conversely, if we have a categorical set we can `forget' the identity morphisms, leaving us with a bare set of objects). In the literature, this sort of groupoid is also known as a `discrete category'.} Each distinct model in the category is thereby taken to represent a different possible system.  
Second, a groupoid is a \textit{categorical equivalence relation} $\mathcal{E}$ if and only if between any two objects $A$, $B$ in $\mathcal{E}$ there is at most one morphism in the morphism set hom$(A, B)$.\footnote{Thus, every categorical set is a categorical equivalence relation, but not vice versa.} For categorical equivalence relations, clusters of objects connected by morphisms fall into the same representational `equivalence class', and thereby represent the same system. 
Finally, we will also consider groupoids that are neither categorical sets nor categorical equivalence relations. Like the latter, such groupoids contain disconnected clusters of representationally equivalent objects, but in this most general case, there may be multiple morphisms between (not necessarily distinct) objects within each cluster---these morphisms represent the different ways in which representationally equivalent models can be transformed into one another.

\centering
\subsection{Relations between theories}
\label{relthry}
\justify

Having represented a scientific theory as a category, one can then treat a theory as being itself an object of a category: the category of theories. A simple category that lends itself to this purpose is the $2$-category \textbf{Cat} whose objects are $1$-categories, whose $1$-morphisms are functors between categories, and whose $2$-morphisms (morphisms between $1$-morphisms) are natural transformations between functors. 
We now explain how to use the morphisms of this category to describe equivalence and inequivalence relations between theories.

\centering
\subsubsection{Equivalence}
\label{equiv}
\justify
Just as we previously used morphisms within a 1-category to describe the standard of sameness for models within a theory, we can use the morphisms of \textbf{Cat}, namely functors, to describe a standard of sameness for \textit{theories}. Two options emerge for such a standard: (a) \textit{categorical isomorphism}; or (b) \textit{categorical equivalence}, which includes isomorphism as a special case. Two (categorically represented) scientific theories $\C_{T_1}$ and $\C_{T_2}$ are \textit{categorically isomorphic} if and only if there exist functors $G: \C_{T_1} \to \C_{T_2}$ and $H: \C_{T_2} \to \C_{T_1}$ such that $GH = 1_{\C_{T_2}}$ and $HG = 1_{\C_{T_1}}$.
And they are \textit{categorically equivalent} if and only if $GH \cong 1_{\C_{T_2}}$ and $HG \cong 1_{\C_{T_1}}$; notice that these `quasi-inverse' functors weaken the equalities in the definition of (a) to natural isomorphisms.\footnote{Where needed we use `$\cong$' to denote isomorphism (both `natural isomorphism' as in the above relations between functors, as well as categorical isomorphism---the meaning will be clear from the context) and `$\simeq$' to denote categorical equivalence.} 

We will eventually take up the question of which of these standards best captures equivalence and inequivalence relations between (categorically represented) scientific theories such as gauge theory. 
For the moment, let us note that mathematical practice recommends (b) for methodological reasons: mathematicians generally take it that we should not limit the expressive resources of our standard of sameness by insisting on categorical isomorphism (thereby assuming that all relevant $2$-morphisms are trivial).
This is motivated by the fact that there are cases of categorical equivalence which are not categorical isomorphism.
For example, in the case of `Gelfand duality' one has a category $\C_G$ of geometric models and a non-isomorphic category $\C_A$ of algebraic models: although the objects of these respective categories look very different, they are mutually interdefinable (albeit only up to isomorphism), and so we would intuitively still want to say that $\C_G$ and $\C_A$ are `structurally the same'; the quasi-inverse functors of categorical equivalence (but not the inverse functors of categorical isomorphism) have precisely the weakness that allows us to say this.

In the rest of this paper, we will be concerned with a rather more trivial case of a categorical equivalence that is not an isomorphism. 
The general form of the case is as follows: (Skel): Given a groupoid $\C$ (and assuming the axiom of choice), one can construct a \textit{skeleton} category, $\text{sk} \C$, by picking an object from each isomorphism class of $\C$ and defining $\text{sk} \C$ to be the full subcategory of $\C$ containing all and only these objects. 
This skeleton is unique up to categorical isomorphism \citep[p. 43]{riehl16}.
It is easy to see that $\C$ is categorically equivalent to $\text{sk} \C$, because the relevant interdefinability relations are almost trivial.\footnote{Note that while the algebraic data of a category and its skeleton are equivalent, geometric properties (such as smoothness) are not in general preserved when one passes to the skeleton.} 

Armed with these conceptual tools, we can now pause to reflect on the question of what is gained by representing theories as groupoids, rather than as sets, of models.
Suppose that we represent our theory as a categorical set of models, and that we choose (a)---i.e. categorical isomorphism---as our standard of sameness. Then, since the property of being a categorical set is preserved by (a), and categorical sets and sets are interdefinable, nothing has been gained by passing from sets to categorical sets. 
On the other hand, choosing (b)---i.e. categorical equivalence---as our standard of sameness reveals a novel feature of the categorical approach, as demonstrated by the following special case of (Skel): Every categorical equivalence relation is equivalent to its skeleton, which is a categorical set \citep[p. 15]{bl05}.  
This fact shows us that the property of `being a categorical set' is not invariant under (b); thus, passing from sets to categorical sets allows us to take advantage of the representational flexibility afforded by (b).\footnote{Turning this point on its head, one might also use it to justify categorical equivalence as a standard of sameness: if one sticks to categorical isomorphism---which does preserve the property of being a categorical set---then one loses the distinctive flavor of a categorical formulation.}
For instance, in Section 3, we will consider a gauge theory example in which formulating the models as a categorical set allows us to pass to a categorically equivalent groupoid (a categorical equivalence relation) of models which is not a categorical set.  

The distinction between the standards of sameness (a) and (b) will be crucial for the work that we wish to do in this paper.
In Section 3, we will deploy this distinction to demonstrate that the intuitive notion of `surplus' is in fact ambiguous, and the resolution of this ambiguity turns on whether one adopts (a) or (b). 
Furthermore, Section 4 argues that we have more than merely methodological reasons to adopt (b): Within the context of gauge theory, (b) is part of an explanation of why `surplus' can play the representational role that it does.

Let us stress that all that we have said so far concerns when two (categorically represented) theories are \textit{formally} equivalent to one another. 
The further question of when these theories are equivalent \textit{tout court} (or `theoretically equivalent') clearly turns on how these formal representations should be \textit{interpreted}.
We do not wish to take a stand on `how to interpret physical theories' or on sufficient conditions for theoretical equivalence here.
For our purposes it suffices to note that, in addition to being related by a formal standard of equivalence (e.g. (a) or (b)), two equivalent theories would at least need to be `empirically equivalent', which means that the objects (models) of the formally equivalent categories should represent the same empirical data (i.e. the empirical data that would be accessible if the physical state of affairs represented by the models were actual). 
In the context of the restricted $U(1)$ gauge theory that we discuss in Section \ref{gaugecat} of this paper, we follow the physicists' standard practice of identifying the empirical content of a field model with its force field $F$.
As we will describe below, the assumption of empirical equivalence places an additional constraint on the functors that relate two categorical representations of $U(1)$ gauge theory. 

\centering
\subsubsection{Inequivalence}
\label{inequiv}
\justify

Finally, there are cases in which it is fruitful to consider relations between theories that have different amounts of structure; in the context of this paper, such a scenario is provided by our analysis of `surplus structure'. By hypothesis, there will be no categorical equivalence between such theories, but one might still try to quantify such differences by means of functorial relationships between theories.

Within the mathematical literature, there is a long tradition of describing such relationships by means of `forgetful functors' between categories. 
More recently, \citet{baez04} have offered a more refined taxonomy of `forgetfulness' in order to facilitate applications to categorical algebraic topology; we will find it convenient to apply this taxonomy to the analysis of `surplus' within our less sophisticated setting.\footnote{One reason for this convenience is that it will allow us to compare our approach to \citep[pp. 1042-3]{W2}, which also borrows Baez et al.'s taxonomy in order to discuss `gauge'.}
In order to understand Baez et al.'s taxonomy, first recall some textbook terminology concerning a functor $G: \C\to\D$: 
\begin{itemize}
    \item $G$ is \textit{full} if and only if for all objects $A, B$ in $\C$, the map (induced by $G$) from hom$(A,B)$ to \sloppy hom$(G(A), G(B))$---i.e. $G_{A \rightarrow B}: \text{hom}(A,B) \to \text{hom}(G(A),G(B))$---is surjective. 
    \item $G$ is \textit{faithful} if and only if for all $A, B \in \C$, $G_{A \rightarrow B}$ is injective.
    \item $G$ is \textit{essentially surjective} if and only if for every object $X$ in $\D$, there is an object $A$ in $C$ such that $G(A)$ is isomorphic to $X$. \fussy
\end{itemize}
There exists a $G: \C \to \D$ which is full, faithful, and essentially surjective if and only if $\C$ and $\D$ are categorically equivalent.\footnote{Notice that the left-hand-side of this biconditional is an existential claim. That there exists such a functor does not entail that every functor between the categories is full, faithful, and essentially surjective. Equivalently, that there exists a forgetful functor between the categories does not entail that they are not categorically equivalent.} 

Now, {\'a} la Baez et al., if $G$ fails to have at least one of these properties then $G$ is a \textit{forgetful} functor. If $G$ fails to be full then it forgets \textit{structure*}, if it fails to be faithful then it forgets \textit{stuff}, and if it fails to be essentially surjective then it forgets \textit{properties}.\footnote{Note that we wish to reserve the term `structure' for the intuitive or pre-theoretical notion of structure, and so we introduce `structure*' as a term of art that replaces Baez et al's use of `structure' to refer to the functorial notion. Thus, when we speak of a functor $G: \C \to \D$ forgetting \textit{structure*}, we are not claiming that this captures the intuitive idea that $\C$, or the scientific theory represented by $\C$, has `more structure' than the theory represented by $\D$. We leave open the possibility that a functor that forgets \textit{stuff} or \textit{properties} may also demonstrate that the domain category has more `structure' in some intuitive sense than the codomain category.} 
Thus, relative to a choice of functor, when comparing two categories $\C$ and $\D$ we can say that if $G$ fails to be full then $\C$ has more \textit{structure*} than $\D$; if $G$ fails to be faithful then $\C$ has more \textit{properties} than $\D$; and if $G$ fails to be essentially surjective, then $\C$ has more \textit{stuff} than $\D$.

\centering
\section{Gauge Theory as a Category}
\label{gaugecat}
\subsection{Gauge theory on contractible manifolds}
\justify

When physicists speak of `gauge theory', they typically have in mind Yang-Mills type gauge theory, of which the prototype is $U(1)$ gauge theory (also known as `electromagnetism'). We have adopted this nomenclature in our paper, and we shall only address the $U(1)$ case for simplicity. Furthermore, we will treat spacetime as a smooth manifold throughout this paper, since none of our points turn on the presence of spacetime metric structure.

As we mentioned in Section 1, a gauge theory is a field theory, and fields are typically taken to possess `locality', in the sense that field configurations on a manifold will be determined by specifying configurations on a collection of local subregions that `cover' that manifold. 
However, in this section, we will only consider a mock-up of field theory that makes two simplifying assumptions. 
First, the theory does not represent any facts about locality, and so a model will simply be a field configuration on the whole of a spacetime. 
Second, we will only consider models that are set on contractible manifolds.
Both of these assumptions will be removed in Section \ref{gaugefunc}. 

What are the fields of a $U(1)$ gauge theory?
Physicists typically formulate the theory in terms of gauge fields, i.e. 1-forms $A\in\Omega^1 (M, \u(1) )$.\footnote{For an introduction to gauge fields, see \citep[Chapter 1]{pokorski20}.}
From a gauge field $A$, one can define the force field $F:=dA$, which is invariant under the gauge transformations:

\beq\label{u1trans}
A \mapsto A + g^{-1} dg
\eeq

\noindent where $g: M \rightarrow U(1)$.\footnote{Throughout this paper we use the mathematicians' convention; the typical physicists' convention would denote the gauge transformation as $A\mapsto A -ig^{-1}dg$.}

Note that by definition, the force field $F$ satisfies the source-free equation of motion $dF = 0$.\footnote{There is also a dual source-free equation of motion $d *F = 0$, but we will not use it in this paper.}
On a contractible manifold $M$ (and assuming that we have no interest in subsystems on non-contractible subregions of $M$) the empirical content of the theory can be expressed solely in terms of the gauge-invariant fields $F$; whence familiar worries about `surplus' when the theory is formulated in terms of the gauge fields $A$.
 
We now turn to the articulation of a $U(1)$ gauge theory on $M$ as a category.
Since such a theory represents states of affairs on a particular spacetime $M$, we will adopt the standard practice of fixing $M$, and taking a model to be a field configuration on $M$.\footnote{For more details of the standard modelling practice, see \citep[Section 2]{BSS}. In Section 4 we will discuss how to assign field configurations to varying manifolds.}
Thus, a category of models will be a category whose objects are such gauge fields, and whose morphisms (if we should choose to include them) are symmetries of the relevant fields.

A specific category of this kind that has been widely used in the mathematical physics literature is the groupoid of gauge fields $\C_{A}$.
$\C_{A}$ is a groupoid---that is not a categorical set or equivalence relation---whose objects are the gauge fields $A\in\Omega^1(M,\u(1))$, and whose morphisms are the gauge transformations given by (\ref{u1trans}).

\centering
\subsection{Other candidates for representing $U(1)$ gauge theory}
\justify

Although $\C_{A}$ is one of the standard categorical representations of gauge theory within physics, it has scarcely been discussed within the philosophical literature.\footnote{For some exceptions, see (\citeauthor{dougherty16}, forthcoming) and \citeauthor{Teh2017}, forthcoming)}
On the other hand, several alternative candidate categories for representing gauge theory have been discussed by \citet{W1, W2}, in conjunction with a particular proposal for understanding `surplus'.
We now consider (simplified versions of) these alternative categories and proceed to develop an analysis of `surplus' which identifies $\C_A$ as possessing the relevant `surplus' structure; we then briefly contrast our proposal with that of \citet{W2}.\footnote{This simplification derives from omitting spacetime isomorphisms, since for our purposes we focus exclusively on the field structure of $U(1)$ gauge theory, which, in this section, is distinguished from its metric structure by keeping the background manifold and metric fixed and identifying models with \textit{n}-forms (or constructions thereof) and the morphisms with the maps between them.}

Alternative representations of gauge theory are given by the following groupoids:
\begin{itemize}    
\item $\mathcal{S}_F$ is the categorical set whose objects are force fields $F\in\Omega^2(M,\u(1))$ and all of whose morphisms are identities.\footnote{Recall two points: First, since $F$ is a force field, it is closed, i.e. $dF = 0$. Second, since we are assuming that $M$ is contractible, Poincare's Lemma implies that for every $F$ there is a non-unique $A\in\Omega^1 (M, \u(1) )$ such that $dA=F$.}

\item $\S_{A}$ is the categorical set whose objects are gauge fields $A\in \Omega^1(M,\u(1))$, and all of whose morphisms are identities.

\item $\mathcal{S}_{[A]}$ is the categorical set whose objects are (gauge) equivalence classes $[A]$ and all of whose morphisms are identities.

\item $\mathcal{E}_{A}$ is the categorical equivalence relation whose objects are gauge fields $A\in \Omega^1(M,\u(1))$ and whose morphisms are given by the gauge transformations $A \mapsto A + \lambda$, where $\lambda$ is a closed $1$-form, i.e. $d\lambda = 0$ (thus, we know from (\ref{u1trans}) that $\lambda = g^{-1} dg$ for some $g$).\footnote{The definition of $\E_{A}$ comes from the category $\overline{\mathbf{EM}}_2$ in \citep[pp.1082-3]{W1}, but where we have omitted spacetime isometries since these will not be relevant for our purposes.} Note that it follows from this definition that the only automorphism of $A$ is the trivial one, i.e. $\lambda=0$. 
\end{itemize}
\noindent We note that while the physically-minded reader may rightly chafe at the definition of $\E_A$ (why arbitrarily omit from (\ref{u1trans}) all automorphisms but the trivial one?), we will provide a mathematical explanation for this choice in Section 5, albeit one that will be at odds with our (physical) motivations.  

In order to give the reader some visual intuition for these definitions, Figure \ref{fig:cats} provides an illustration of these groupoids and $\C_A$. Notice that boxes (a)-(c) and (e) only contain identity automorphisms, whereas box (d) also contains non-trivial automorphisms.\footnote{Explicitly, the non-trivial automorphisms contained in $\C_A$ correspond to shifts by a (non-zero) constant function, which in turn give rise to the `rigid' or `global' $U(1)$ symmetries of electromagnetism, i.e. the (non-trivial) constant maps $g: M\to U(1)$. It is now clear why physicists typically take $\C_A$ to be the `physical category': when one couples $U(1)$ gauge theory to the matter sector (e.g. scalar or fermion fields), these rigid transformations are used to construct non-trivial conserved charges by means of Noether's (first) theorem. By contrast (and putting aside boundary terms), only trivial charges can be associated with the non-constant maps.} 
This last feature explains why box (c) contains exactly two arrows between distinct objects in the same cluster, whereas box (d) contains many such arrows (though not explicitly depicted, this is implied by concatenating non-trivial automorphisms with the depicted arrows).

\vspace{1mm}
[Insert Figure 1 about here]
\vspace{1mm}

We are now in a position to consider the functorial relationships between these different theories. 
In light of our discussion of empirical equivalence at the end of Section \ref{equiv}, we will only consider functors that satisfy the following criterion:

\begin{quote}
Empirical Equivalence of $U(1)$ Theories: Fix a contractible manifold $M$. We will say that two categorical representations $\C$ and $\D$ of $U(1)$ gauge theory on $M$ are \textit{empirically equivalent} if and only if for each model $X \in \C$ (which can be a gauge field, equivalence class of gauge fields, or force field depending on the category in question) with a force field $F$ (either induced by $F:=dA$, or explicitly defined in the model) there is a model $Y \in \D$ with the same $F$, and vice versa for each $Y \in \D$. 
We thus say that a pair of functors $G:\C \to \D$ and $H:\D \to \C$ \textit{realises an empirical equivalence} if and only if for every model (object) $X$ in $\C$ there is a model $Y$ in the image of $G$ with the same force field $F$, and vice versa for every model $Y \in \D$ and the image of $H$.
\end{quote}

\noindent 
In other words, such functors must act in one of the following ways on the categories in Figure 1: they take a gauge field $A$ to a force field $F$ such that $F=dA$ (and vice versa); they take an equivalence class of gauge fields $[A]$ to a gauge field $A$ such that $A \in [A]$ (and vice versa); and finally, they take an equivalence class of gauge fields $[A]$ to a force field $F$ such that $F=dA$ for all $A\in [A]$ (and vice versa).

For our purposes, it will suffice to consider the following empirical equivalence-realising functors:\footnote{Note that the actions on the morphisms by $G, H, I, K$ and $M$ are determined by the definition of a functor since the only morphisms in the domain categories are identities.}

\begin{quote}
\begin{description}

\item[$G$:] $\S_{A} \to \mathcal{S}_{[A]}$ such that for all $A$ in $\S_{A}$: $G(A)=[A]$ where $A\in [A]$; identities are sent to identities.

\item[$H$:] $\mathcal{S}_F \to \mathcal{S}_{[A]}$ such that for all $F$ in $\mathcal{S}_F$: $H(F)=[A]$ where $F=dA$ for all $A \in [A]$; identities are sent to identities.

\item[$I$:] $\mathcal{S}_{[A]} \to \mathcal{E}_{A}$ such that for all $[A]$ in $\mathcal{S}_{[A]}$: $I([A])$ = a representative $A$ of $[A]$; identities are sent to identities.\footnote{Note that this functor lacks naturality (in the choice of the representative of the equivalence class). The same applies to $K$ and $M$.}

\item[$J$:] $\mathcal{E}_{A} \to \mathcal{S}_{[A]}$ such that for all $A$ in $\mathcal{E}_{A}$: $J(A) = [A]$ where $A\in[A]$; and for all $A, B$ in $\mathcal{E}_{A}$ such that $J(A)=J(B)$: $\text{hom}(J(A),J(B)) = \lbrace 1_{J(A)}\rbrace$.

\item[$K$:] $\mathcal{S}_{[A]} \to \C_{A}$ such that for all $[A]$ in $\mathcal{S}_{[A]}$: $K([A]) = $ some $A$ in $\C_{A}$ such that $A \in [A]$; identities are sent to identities.

\item[$\tau$:] $\C_{A} \to \mathcal{S}_{[A]}$ such that for all $A$ in $\C_{A}: \tau(A) = [A]$ where $A \in [A]$; and for all $A, B$ in $\C_{A}$ such that $\tau(A)=\tau(B), \text{hom}(\tau(A), \tau(B)) =\lbrace 1_{\tau(A)}\rbrace$.

 \item[$M$:] $\mathcal{S}_{[A]} \to \S_{A}$ such that for all $[A]$ in $\mathcal{S}_{[A]}$: $M([A]) = $ some $A$ in $\S_{A}$ such that $A \in [A]$; identities are sent to identities.

\end{description}

\end{quote}

Since these functors realise empirical equivalences, the categories in Figure 1 are all empirically equivalent to each other. We now turn to the further inter-theoretic relationships encoded by these functors, as depicted in Figure 2.

\vspace{1mm}
[Insert Figure 2 about here]
\vspace{1mm}

First, there are (formal) equivalences, as represented in the vertical axis of Figure 2, where $\cong$ denotes categorical isomorphism and $\simeq$ denotes categorical equivalence.
It is easy to see that $H$ and its inverse demonstrate that $\mathcal{S}_F$ and $\mathcal{S}_{[A]}$ are isomorphic (cf. \citep[p. 1082]{W1}).
Similarly, $I$ and $J$ are quasi-inverse functors which demonstrate that $\mathcal{S}_{[A]}$ and $\mathcal{E}_{A}$ are categorically equivalent (but non-isomorphic). 

Notice that this equivalence is simply a special case of what we called (Skel) in Section \ref{equiv}, and in particular, of the fact that a categorical equivalence relation ($\mathcal{E}_{A}$) is equivalent to a categorical set ($\mathcal{S}_{[A]}$ and $\mathcal{S}_F$). 

Next, there are inequivalences, as represented in the horizontal axis of Figure 2.
The following propositions show that $G, K, \tau$ and $M$ are all forgetful functors which encode differences in the structure of theories $\C_A$, $\S_{[A]} (\cong \S_F \simeq \E_A)$, and $\S_A$ (see the Appendix for the proofs):

\begin{restatable}{proposition}{G}
$G$ forgets (only) \textit{structure*}. So, with respect to $G$, $\S_{A}$ has more \textit{structure*} than $\mathcal{S}_{[A]}$.

\end{restatable}

\begin{restatable}{proposition}{K}
$K$ forgets (only) \textit{structure*}. So, with respect to $K$, $\mathcal{S}_{[A]}$ has more \textit{structure*} than $\C_{A}$.
\end{restatable}

\begin{restatable}{proposition}{T}
\label{CMstuff}
$\tau$ forgets (only) \textit{stuff}. So, with respect to $\tau$, $\C_{A}$ has more \textit{stuff} than $\mathcal{S}_{[A]}$.
\end{restatable}

\begin{restatable}{proposition}{M}
\label{M}
$M$ forgets (only) \textit{properties}. So, with respect to $M$, $\mathcal{S}_{[A]}$ has more \textit{properties} than $\S_{A}$.
\end{restatable}

In the next subsection we discuss how these forgetful relationships correspond to the notion of `surplus' structure.

\centering
\subsection{Surplus and inter-theoretical comparisons}
\justify

Within philosophy, the terms `surplus' structure or `descriptive fluff' have been used to refer to features of gauge theory that seem to play no representational role in theory (cf. \citep{Redhead2003, Earman2004}). And although physicists tend not to use either of these terms, it is extremely common for them to claim that `gauge symmetry is a \textit{redundancy}, because it relates gauge fields that represent the same physical state' and that it is the `gauge equivalence classes' that capture the non-redundant physical data (the references for this view are too many to list, but see especially \citep[Section 6.1.1]{tong06}, \citep[p.64]{AG2011}, and \citep[p.497]{sundermeyer14}). 
Thus, the notion of `surplus' structure or `redundancy' that we are after applies to theories in which (i) gauge fields corresponding to the same $F$ are represented as representationally equivalent, and (ii) the gauge symmetries are included as part of the structure that is `surplus'.
While in Section 4 we argue that this notion of `surplus' \textit{does} in fact play a representational role when one adopts a sufficiently rich conception of `gauge theory', for now we will only be concerned with the question of how precisely to understand `surplus' within a simplified setting.
As a rough first pass, we might say that a theory with `surplus' structure is one that includes the gauge symmetries---understood as isomorphisms---as part of the theory (thus satisfying (i)), and whose structure is in some still-to-be-precisified sense `surplus' with respect to a theory that only includes the gauge equivalence classes (thus satisfying (ii)).
In other words, such a theory satisfies
\begin{quote}
(Surplus): A (gauge) theory has `surplus structure' (which includes the gauge fields and gauge symmetries) over and above its gauge equivalence classes.\footnote{Note that this is essentially the formulation given in \citep{Redhead2003}.} 
\end{quote}

\noindent 

How can this intuitive notion of (Surplus) be precisified by means of categorical tools?
As we see it, the `theories as categories' paradigm allows us to provide a four-fold precisifiation of what it means for a theory to have `structure over and above its gauge equivalence classes'.

First, we use categories to represent the two theories that we are comparing: a theory whose field content involves only gauge equivalence classes, and a theory that is a candidate for having (Surplus) structure. We thus take $\S_{[A]}$ to represent the former, whereas the latter can be represented by either $\C_A$, $\E_A$, or $\S_A$. 

Second, our target concept of `surplus' applies to theories that use gauge symmetries to encode relations of representational equivalence between gauge fields. 
Since $\S_A$ does not include these symmetries, it is not the appropriate category for capturing the content of (Surplus). 
Consequently, our focus will be restricted to comparisons between $\S_{[A]}$ (the theory that does not satisfy (Surplus)) and either $\C_A$ or $\E_A$ (the candidate theories that satisfy (Surplus)).

Third, once we have represented theories as categories, we will need to settle on a standard of sameness for theories (thus disambiguating an aspect of what is meant by `over and above' in (Surplus)). 
If categorical isomorphism is our standard, then we obtain what one might think of as the `naive' reading of (Surplus), on which $\E_A$ counts as having `surplus' over and above the gauge equivalence classes of $\S_{[A]}$.
However, as we argued earlier, there are strong methodological reasons for picking categorical equivalence as our standard.
We will thus make this choice, and so $\mathcal{E}_{A}$ will count as `the same theory' as $\mathcal{S}_{[A]}$.
Consequently, the only candidate category which satisfies (Surplus) is $\C_A$.

Fourth, we perform a functorial comparison between $\C_A$ and $\S_{[A]}$ that will clarify the `surplus' structure that $\C_A$ possesses over and above $\S_{[A]}$. 
Such a comparison is the content of Prop. 3.2.3 (and the left arm of Figure 2), which tells us that the forgetful truncation functor $\tau: \C_{A} \rightarrow \mathcal{S}_{[A]}$ forgets \textit{stuff} in Baez et al.'s taxonomy of forgetfulness.\footnote{Here, one familiar interpretation of \textit{stuff} as `entities within the objects that can be equipped with structure' does not apply because we are not dealing with categories which are structured sets. However, the `\textit{structure/stuff/properties}' distinctions are relevant beyond the case of structured sets, and are standardly used in categorical algebraic topology, for which we refer the reader to \citep{baez05}. For a general account of truncation, see \citep{lurie}.}

This four-fold precisification of `surplus' structure allows us to define \textit{surplus* structure} as the \textit{stuff} that is forgotten by the functor $\tau$. This provides the following characterisation of what it means for a $U(1)$ gauge theory to have such structure:

\begin{quote}
(Surplus*): A (gauge) theory contains the \textit{stuff} that is forgotten by $\tau$ (namely the non-trivial automorphisms of the gauge fields and the result of concatenating them with the morphisms already included in $\E_A$).\footnote{To see clearly why it is the non-trivial automorphisms that are important, it is useful to compare the skeleton sk$\C_A$---which is categorically equivalent to $\C_A$---to $\S_{[A]}$; the key difference is precisely the presence of the non-trivial automorphisms.}
\end{quote}

It is worth pausing briefly to contrast our notion of `surplus' with the rather different concept that is being analysed in \citep{W2} (albeit also under the heading `surplus'), where we will omit spacetime isometries from the latter's categories since we are only concerned with gauge structure. 
First, note that both accounts hold that the theory without `surplus' (in each respective sense) is $\S_{[A]}$, and both take categorical equivalence to be the correct standard of sameness for theories. 
The key difference between these accounts lies in their respective target notions of `surplus'. On the one hand, we are interested in a notion of `surplus' that is possessed by theories which take gauge fields to be representationally equivalent (and which represent this by means of gauge transformations between gauge fields); thus, $\C_A$ is our candidate for such a theory and `surplus' is characterised by the \textit{stuff}-forgetting functor $\tau: \C_A \rightarrow \S_{[A]}$ (cf. the left arm of Fig. 2).
By contrast, Weatherall's notion of `surplus' applies to a theory that does not represent gauge fields as representationally equivalent, namely $S_A$ (which has no gauge transformations between objects and is essentially a set). Thus, his notion of `surplus' is characterised by the structure*-forgetting functor $G: \S_{A} \rightarrow \S_{[A]}$ (cf. the right arm of Fig. 2), which intuitively `forgets the lack of arrows between objects'.\footnote{
More generally, one way of motivating the focus on $\S_A$ as a candidate category for possessing `surplus' is to adopt a conception of structure that focuses on the structure of \textit{objects} of a category, and uses the morphisms to track how much structure these objects have.
An implicit maxim guiding such an approach is something like: the more isomorphisms in a category, the less (relevant) structure its objects have, because isomorphisms are structure-preserving maps; the more there are, the less structure there is to preserve.}

We take no issue with Weatherall's account as an attempt to compare the objects of $\S_A$ (which is essentially a set) with $\S_{[A]}$. 
However, from the point of view of physical practice, there are at least two good reasons to focus on $\C_A$ and its comparison with $\S_A$.
First, and as we mentioned earlier (cf. fn. 26), $\C_A$---but not $\S_{[A]}$ or $\S_A$---contains the global gauge transformations (constant maps $g: M \rightarrow U(1)$) which are responsible for (non-trivial) charge conservation when a $U(1)$ gauge theory is coupled to matter fields. 
Second, and as we will argue in the next section, the surplus* structure of $\C_A$ is required in order to explain the oft-cited conceptual connection between the `redundancies' of a gauge theory, on the one hand, and the `locality' of the theory, on the other. 

\centering
\section{Gauge Theory as a Functor}
\label{gaugefunc}
\subsection{(Richness) and (Locality)}
\justify

In this section, we return to the puzzle that we raised in the introduction: how can `surplus' structure be an essential feature of a theory?
Roughly speaking, our strategy for resolving this puzzle will be to argue that in order to represent `locality' (i.e. a property of field theory, as it is standardly conceived) the notion of theory in Section 3 needs to be extended to a more general context; from the perspective of this generalisation, surplus* structure is required in order to represent a rich set of physically relevant models that possess locality. 

More carefully now, recall that the formulation of a `theory as a category' in Section 3 rested on two simplifying assumptions: (i) we only considered fields set on contractible regions; and (ii) we ignored `locality' by taking such regions to be entire spacetimes. 
Recall too that our precisification of `surplus', i.e. surplus*, was predicated on this narrow notion of a `theory'.
In Sections 4.1.1 and 4.1.2, we will consider how removing (i) and (ii) suggests a move to a broader \textit{functorial} notion of `theory', which is not merely a collection of models, but rather a uniform assignment of collections of models to different spacetimes (including subregions of a spacetime).

The motivation for lifting (i) is:
\begin{quote}
(Richness): Gauge theories have a rich set of physically relevant `topological soliton' solutions, which are defined on non-contractible manifolds. 
\end{quote}

\noindent The use of the term `physically relevant' here covers the spectrum of empirical and theoretical roles that such solutions play in contemporary physics.\footnote{See \citep{Manton2004} for an overview of (classical) topological soliton solutions.} On the empirical side: many solitons (instantons in Yang-Mills theory; vortices in condensed matter physics) are used to make predictions about models of field theory.
And on the theoretical side, solitons are even more pervasive: for instance, many field-theoretic dualities make essential use of soliton solutions. 
Thus, any philosophically adequate analysis of gauge theories must take into account the richness of their solitonic models.
Since the `theories as categories' of Section 3 were only defined on contractible manifolds, they do not in themselves have the resources to capture (Richness); however, Section 4.1.1 will introduce the machinery of principal $G$-bundles that allows us to do so. Furthermore, Section 4.1.2 will introduce a notion of `gluing' that forges a link between the `theories as categories' of Section 3 and (Richness). 

The motivation for lifting (ii) is the fact that field theories are typically taken to possess:
\begin{quote}
(Locality): All physically relevant global models (i.e. field configurations on an arbitrary spacetime $M$) are determined by `compatibly gluing' a set of local models (i.e. field configurations on contractible subregions $U_i$ of $M$, where $\{U_i\}$ is a cover of $M$).
\end{quote}

\noindent
In other words, (Locality) is a constraint that relates the physically relevant global models on $M$ to a family of `theories as categories' (whose objects are local models), where the family is parameterised by contractible subregions of $M$. 

By itself, (Locality) does not tell us what the physically relevant global models are. However, (Richness) helps specify this class of models, and as argued below, when (Richness) and (Locality) are conjoined they imply (Surplus*), thus showing how `surplus' structure can be an essential feature of theories that are rich enough to combine three aspects: one can represent global models (i.e. topological solitons); one can represent local models (i.e. models of the sort considered in Section 3); and one can represent the relationship between local and global models (i.e. how the former can be combined to yield the latter). 

In order to mount this argument more precisely, we will now proceed to describe the categorical formulation of (Richness) and (Locality).
While this is a subject matter about which one can be exceedingly abstract, we will focus on a concrete and elementary example for $U(1)$ gauge theory: the Dirac monopole on the sphere $M=S^2$.

\centering
\subsubsection{Representing (Richness)}
\label{richness}
\justify

The global `topological soliton' models that exemplify (Richness) are standardly defined by means of the machinery of a principal $G$-bundle with connection.\footnote{A principal $G$-bundle is a fibre bundle $\pi: P\to M$ with a continuous right action by the group $G$ on the total space $P$ which preserves fibres.} 
For instance, in the our focal case of $U(1)$ gauge theory, (Richness) requires that we take into account the Dirac monopole model, which is a $U(1)$-bundle $P$ equipped with a connection $1$-form $\omega$ over the sphere $S^2$.
Such global models can be understood as combining two pieces of data: (i) a topological charge that corresponds to the isomorphism class of $P$; and (ii) a particular connection $\omega$ on $P$ (up to connection-preserving bundle isomorphism).
In the Dirac monopole case, it will be particularly useful for us to view the monopole's global force field $F \in \Omega^2 (S^2, \u(1))$ as a summary of (i) and (ii): $F$ is constructed from the bundle's curvature $d\omega$, and by a standard result from Chern-Weil theory, the cohomology class of $F$ corresponds to the topological charge of the monopole (and thus to the isomorphism class of $P$).\footnote{Here we note a technical subtlety: The first (integral) Chern class of a $U(1)$-bundle on $M$ is a complete invariant of the bundle's isomorphism class, but when we represent it by $F$, we are omitting possible torsion information in this class. In the case where the connection is flat, such torsion data corresponds precisely to the monodromies that account for the Aharanov-Bohm effect, cf. footnote \ref{footnoteAB}.}  

Due to a classical result by \citet{Steenrod}, we know that isomorphism classes of principal $G$-bundles over $M$ are in bijection with homotopy classes of maps from $M$ into the classifying space $BG$ of the group $G$.
Thus, in cases where the classifying space $BG$ is well-understood, homotopy classes of maps provide a useful classification of the possible topological charges of a soliton model.
Applying this result to this case of the Dirac monopole, we see that the principal $U(1)$-bundles $P$ over $S^2$ are in bijection with $[S^2, BU(1)=\mathbb{C}P^\infty]=\pi_2(\mathbb{C}P^\infty)\cong \mathbb{Z}$.
Thus, the topological charge of a Dirac monopole model is $n \in \mathbb{Z}$. The charge is also known as the `first Chern number' of the model since it can be computed as $n =\frac{1}{2\pi} \int_{S^2} F$, where $F$ represents the first Chern class of $P$. 

In order to speak of a collection of global models on $M$, it is standard practice to define the functor $Bun: \textbf{Man}\to \textbf{Grpd}$, which assigns to any $M$ (not necessarily contractible) the groupoid $Bun(M)$ whose objects are principal $U(1)$-bundles over $M$ with connection $(P,\omega)$, and whose morphisms are connection-preserving bundle isomorphisms over the identity on $M$.
In particular, $Bun(S^2)$ contains all the Dirac monopole models of $U(1)$ gauge theory.
We note that while one could say that we are still within the `theories as categories' paradigm since $Bun(M)$ is a category (for a fixed $M$), this observation misses the point of the functorial perspective on (Richness): our theory assigns global models to arbitrary $M$ (and regions of $M$) in a uniform manner.\footnote{The relevant sense of `uniformity' here includes the preservation of the pattern of morphisms between objects in the functor's domain. See also the discussion of `uniformity' and its relationship with `canonicity' in Section 3.4 of \citep{tt2016}.} 

What is the relationship between $Bun(M)$ (and thus the Dirac monopole models) and the previously mentioned groupoid of gauge fields $\C_A$ (on $M$)? 
Such a relationship is easy enough to specify in the case where $M$ is contractible: $Bun(M)$ and $\C_A$ are then equivalent as categories.\footnote{This equivalence works by mapping $A\in \Omega^1(M,\mathfrak{u}(1))$ to the trivial bundle $M\times U(1)$ with connection $\omega:={p_1}^*A+ {p_2}^*\mathfrak{m}$, where $\mathfrak{m}\in \Omega^1(U(1),\u(1))$ is the Maurer-Cartan form, and $p_1:M\times U(1)\to M, p_2:M\times U(1)\to U(1)$ are projection maps. (Explicitly, the Maurer-Cartan form is obtained by taking for any $g\in U(1), \mathfrak{m}_g(v):=(L_g)_*^{-1}(v)$, where $v\in T_g U(1)$, and $L_g:U(1)\to U(1)$ is left multiplication by $g$.)
Since every connection on $M\times U(1)$ is of this form, we have the desired bijection between $\Omega^1(M, \u(1))$ and connection forms on \textit{trivial} principal $U(1)$-bundles.
Furthermore, since \textit{all} principal bundles over contractible $M$ are trivial, we have the desired categorical equivalence between $Bun(M)$ and $\C_A$.}
However, for non-contractible $M$, the relationship between $Bun(M)$ and $\C_A$ becomes more complicated, because the groupoid of gauge fields $\C_A$ is only defined on contractible manifolds.\footnote{One could consider defining $\C_A(M)$ for non-contractible $M$, but this would depend on the cohomology of $M$ and in general could give only the trivial solution; in other words, in general there is no non-trivial global gauge field defined on non-contractible $M$.} 
And since the Dirac monopole model requires a non-contractible $M= S^2$, we will need to forge a more subtle relationship between $Bun(M)$ and $\C_A$ if we are to connect the `theories as categories' paradigm---and thus (Surplus*)---to the Dirac monopole models which exemplify (Richness).
Representing such a relationship is required by (Locality), which we now proceed to discuss.

\centering
\subsubsection{Representing (Locality)}
\label{loc}
\justify

Recall that (Locality) requires that every physically relevant global model on a (possibly non-contractible) manifold $M$ is determined by `compatibly gluing' some family of local models parametrised by the contractible cover $\{ U_i \}$ of $M$.
For $U(1)$ gauge theory, a local model on $U_i$ is contained in one of the `theories as categories' of Section 3.
Thus, from a categorical point of view, the most natural way to conceptualise this `family of categories' parametrised by $\{ U_i \}$ is by representing theories as \textit{functors} that assign the categories of Section \ref{gaugecat} to contractible subregions $U_i\subseteq M$.\footnote{Below we will only describe the action of such functors on objects; for the issue of how such functors are defined on morphisms, we refer the reader to \citep{BSS}.} 
Functoriality is important here because the `theory functors' capture the relationship between `relations between regions of spacetime' and `possible gluing relations between models on those regions'.\footnote{We note that one could again claim to be describing this data within the `theories as categories' framework by fixing spacetime regions and adding the relevant relationships by hand, but this would be nothing other than an implicit appeal to functoriality.}

We will define such functors by extending the `theory as a category' notation of the previous section in the following way. Suppose we previously used $\C_T$ to denote a groupoid of models for $T$ for some fixed contractible manifold $M$. 
Abusing notation, we now use $\C_T: \textbf{Man}^{\copyright}\to \textbf{Grpd}$ to denote a \textit{functor} from the category of contractible manifolds to the category of groupoids, which assigns that groupoid of models to $M$. 
Thus, in this new notation, $\S_F: \textbf{Man}^\copyright \to \textbf{Grpd}$ takes an arbitrary contractible manifold $M$ to the categorical set $\S_F(M)$; \textit{mutatis mutandis} for the functors $\E_A$ and $\C_A$.\footnote{Notice that this `theories as functors' paradigm includes the `theories as categories' paradigm as a special case: to recover the latter, we simply fix a contractible $M$ on which to evaluate the functors.
By highlighting the functorial aspect we want to emphasise the significance of being able to consider varying background spacetimes. It is this feature which is significant when dealing with the locality of a theory.} 
We can now give a preliminary formulation of (Locality) in these terms: it is the requirement that any physically relevant global model in $Bun (M)$ (where $M$ is possibly non-contractible) can be constructed by `compatibly gluing' some collection of local models $\{ A_i \in \C_T (U_i)  \}$, where $\{ U_i\}$ is a contractible open cover of $M$.

To go further, we will need to articulate the notion of `compatible gluing' that relates global models to some collection of local models. 
While there are various abstract and technical ways of precisifying this notion, it will suffice for our purposes to provide an elementary description of gluing for $U(1)$ gauge theory.\footnote{For an elegant formulation of gluing in terms of homotopy limits, see equation (3.4) of \citep{BSS}. This more sophisticated perspective further highlights the importance of the `theories as functors' paradigm for field theory: it allows us to describe a global theory functor as a (higher) limit of a local theory functor.}
Suppose that we have two contractible regions $U_1, U_2$ with non-trivial intersection $U_{12}:=U_1\cap U_2$, and let $A_1 \in\C_T(U_1)$ and $A_2\in \C_T(U_2)$ be a pair of local models over each region.\footnote{For simplicity, we will only deal with a pair of contractible subregions, but the extension to an arbitrary number of subregions will be straightforward. }
We then define the groupoid $\overline{\C_T}(U_{12})$ by taking its objects to be models in $\C_T(U_i)$ (where $i = 1,2$) restricted to $U_{12}$ and its morphisms to be the morphisms in $\C_T(U_i)$ which are similarly restricted.\footnote{Note that this is not the same thing as applying the functor $\C_T$ to $U_{ij}$; the latter will not in general make sense because $U_{ij}$ is not necessarily contractible.} 
We say that the local models $A_1$ and $A_2$ can be `compatibly glued' if and only if there is at least one morphism in $\overline{\C_T}(U_{12})$ between $A_1|_{U_{12}}$ and  $A_2|_{U_{12}}$.  
A choice of gluing is thus a choice of morphism that relates the local models on $U_{12}$; as we will soon see in the case of the Dirac monopole, we will be particularly interested in the homotopy class determined by such a choice.
 
We can now articulate the notion of $U(1)$ gauge theory for which it makes sense to say that `the theory satisfies (Locality)' with respect to the theories of Section 3.
Such an \textit{expanded theory} has three elements: (i) it represents global models by means of a `global theory' $Bun( ~\cdot~)$; (ii) it represents local models by means of a `local theory' $\C_T (~\cdot~)$; and (iii) it represents global-local relationships such as (Locality)---or lack thereof---by means of $\overline{\C_T}(U_{ij})$ for all non-trivial overlapping regions $U_{ij}$.

%\footnote{We note that in order to represent this data, one could simply fix $M$ and some contractible cover $\{ U_i \}$, and consider the collections of models $Bun(M)$ and $Bun(U_i)$ as well gluing relationships between the latter. Superficially, it thus appears that we are back in the `theories as categories' framework; however, all this does is to implicitly assemble the functorial data in a way that does not highlight how the assignments behave `uniformly' under morphisms and changes of base manifold. Furthermore, the functorial provides even more powerful and unifying conception of `locality' once we consider `limits' of functors, as in equation (3.4) of \citep{BSS}. We are grateful to an anonymous referee for encouraging us to be explicit about this point.}

%and consider  take $\C_T$ to be $Bun(~\cdot~)$ evaluated non-contractible manifolds. This would indeed capture both (Richness) and (Locality) without having to explicitly invoke functors, and thus remaining in the theories-as-categories framework . However, doing this requires explicitly specifying each relevant manifold under consideration and loses the unity of grouping all of this data into one theory. What our construction (of considering the theories as functors) does is provide a way in which local gauge fields (from the categories considered in Section \ref{gaugecat}) can be combined to form global fields in way manner which is independent of the manifolds in question. We are grateful to an anonymous referee for encouraging us to be explicit about this.}

We now have a formulation of (Locality) that applies to a general $M$. 
In our subsequent discussion of the Dirac monopole, it will be advantageous to use a less general formulation that makes 
explicit reference to the monopole's global force field $F \in \Omega^2 (S^2, \u(1))$.\footnote{This is less general because it does not suffice to capture Aharanov-Bohm type phenomena, where $F=0$ and so the relevant cohomology classes for the physics stem from non-trivial $1$-cycles in the spacetime. But it will suffice for our purposes, and is at any rate much closer to how physicists intuitively think of the monopole.\label{footnoteAB}} 
We will say that an expanded theory satisfies $F$-(Locality) if and only if the force fields $F$ of any of its physically relevant global models can be constructed by taking some collection of local models---one from each $\C_T (U_i)$---whose fields gives rise to $F|_{U_i}$, and compatibly gluing these over the non-trivial overlaps $\{ U_{ij} \}$. 

It should now be evident that the structure of the morphisms in the local theories $\S_A, \E_A, \C_A$ (now viewed as functors) of Section 3 place stringent constraints on the extent to which an expanded theory can satisfy (Locality).
Without yet introducing any assumptions about what the physically relevant global solutions of a theory should be, let us consider the implications of taking each of these categories to provide a collection of local models.

First, let us consider local models coming from $\S_A$. Recall that each $\S_A(U_i)$ is a categorical set whose objects are fields $A\in \Omega^1(U_i,\u(1))$ and whose only morphisms are the identities. Thus, for any two models $A_1\in \S_A(U_1), A_2\in\S_A(U_2)$ such that $F|_{U_{12}}=dA_1 |_{U_{12}}=dA_2 |_{U_{12}}$, the restrictions $A|_{U_{ij}}, A'|_{U_{ij}}\in \overline{\S_A}(U_{ij})$, can only be compatibly glued by an identity morphism.
In other words, if we wish to use local models to construct a global force field $F$ on $M$, we must be in a situation in which the local gauge fields stem from restricting a global gauge field on $M$. 
But on any manifold with non-trivial de Rham cohomology in degree $2$ (e.g. $M=S^2$), no global gauge field can give rise to a non-trivial force field $F$.
Thus, an expanded theory that links global models to $\S_A (~\cdot~)$ can only satisfy $F$-(Locality) with respect to a trivial global model; in particular, monopoles with non-trivial topological charge cannot be represented in this way. 

On the other hand, in the case of $\E_A$ (respectively $\C_A$), there exist non-automorphism morphisms between $A|_{U_{ij}}$ and $A'|_{U_{ij}}$ in $\overline{\E_A}(U_{ij})$ (respectively $\overline{\C_A}(U_{ij})$).
Thus, these local theories have the following feature: the structure of their isomorphisms is rich enough to compatibly glue \textit{distinct} local gauge fields which nonetheless give rise to an identical force field on $U_{ij}$.
In the next section, we will consider the extent to which this feature allows expanded theories to represent all monopole models while simultaneously satisfying (Locality).

\centering
\subsection{(Richness) and (Locality) imply (Surplus*)}
\justify

Recall that our goal in this section was to explain how `surplus' structure---now precisfied as surplus* structure---can be an essential feature of gauge theory, and our strategy was to advert to a notion of `theory' that is more faithful to physical practice---and thus incorporates (Richness) and (Locality)---instead of the simplistic theories of Section 3.
In such an `expanded theory', as we earlier called it, (Richness) is a property of the global models, `surplus' (however one conceptualises it) is a property of the local models, and (Locality) is a relation between global models and local models.

Our main claim will be that conjoining (Richness) and (Locality) implies that the relevant local theory needs to satisfy (Surplus*), i.e. it should be $\C_A$.
Hence our slogan that (Richness) and (Locality) imply (Surplus*).

We begin by giving a more abstract argument for the converse of our main claim: an expanded theory that has $\C_A$ as its local theory simultaneously satisfies (Richness) and (Locality).
(Richness) tells us that our theory should include as a global model any $(P, \omega) \in Bun(M)$, and so conjoining it with (Locality) yields the requirement that for any $(P, \omega)$, there exists a collection of local models $\{ A_i \in \C_T (U_i) \}$ that can be compatibly glued to yield $(P, \omega)$.
It will be straightforward to see that this requirement is satisfied for $\C_T = \C_A$ if we note that any $(P, \omega)$ has the following `local' description: given some contractible cover $\{ U_i \}$ of $M$, $P$ can be described by patching together the bundle charts $\{ U_i \times U(1) \}$ by means of the `cocycles' $h_{ij}: U_{ij} \rightarrow U(1)$ (for any non-trivial overlap $U_{ij}$); similarly, $\omega$ can be described by patching together $\omega_i:= \omega|_{\pi^{-1} (U_i)}$, where $\pi: P \rightarrow M$ is the bundle projection map.\footnote{More precisely, the prinicpal bundle $\pi:P\to M$ is given by taking the quotient space $P:= (\amalg U_i\times U(1))/\sim$ where the equivalence relation identifies $(x,g)\in U_i\times U(1)$ with $(x,h_{ij}(x)g)\in U_j\times U(1)$. The maps $\phi:\pi^{-1}(U_i)\to U_i\times U(1)$ given by $[x_i,g]\mapsto (x_i,g)$ for $x_i\in U_i$ are the local trivializations of $P\to M$, and the maps $h_{ij}$ are the transition functions between these local trivializations.}
Let $\{ A_i \}$ be a collection of local models (i.e. gauge fields) along with some choice of gluing (i.e. gauge transformations)
\beq
A_j=A_i+g_{ij}^*\mathfrak{m}= A_i + g_{ij}^{-1} dg_{ij}
\eeq
\noindent for all $i,j$, where we have emphasized the role of the Maurer-Cartan form $\mathfrak{m}$.
To see that any $(P, \omega)$ can be reconstructed from some $\{ A_i \}$ along with a choice of gluing, we now make two observations. 
First, by setting $g_{ij}=h_{ij}$ for all $i,j$, our choice of gluing $\{ g_{ij} \}$ yields precisely the cocycle data that defines $P$ (up to isomorphism); in other words, the gauge transformations in $\C_A$ encode the structure that allows us to reconstruct any $P$.
Second, any connection $\omega$ on $P$ can be reconstructed (up to gauge transformations) by choosing $\{ A_i \}$ such that $\omega_i=\pi^*A_i +p_2^*\mathfrak{m}$, where $p_2$ is the projection map $U_i \times U(1) \rightarrow U(1)$.\footnote{For a proof of this reconstruction as well as the general equivalence between $(P,\omega)$ and a gluing of local models, see \citep[Thm 14.4]{sontz}.}

This line of reasoning also indicates why one should expect our main claim to be true.
For suppose that the relevant local theory $\C_T$ does not meet (Surplus*), i.e. it is one of $\S_{[A]}$, $\S_F$ or $\E_A$, and thus does not contain the additional morphisms (\textit{stuff}, cf. Prop. 3.2.3) contained in $\C_A$. 
Then, since the isomorphism class of $P$ is determined by its cocycle data, and since (as we have just argued) the cocycle data is contained in the morphisms of $\C_A$, it is reasonable to expect that it will not be possible to reconstruct an arbitrary $P$ from a local theory whose morphisms are less rich than those of $\C_A$. 
In fact, this is exactly what happens, and we shall now provide a more incisive argument for this conclusion in the case of the Dirac monopole, for which it will be expedient to formulate (Locality) in terms of the monopole's global force field $F$.

In the Dirac monopole case, $M=S^2$, take the open cover of $M$ to be the contractible regions $U_1=S^2\setminus\{s\}$ and $U_2=S^2\setminus\{n\}$, i.e. the sphere with the south and north pole removed, respectively. 
In this context, the conjunction of (Richness) and $F$-(Locality) amounts to the requirement that the global force field $F$ of any monopole model can be determined by compatibly gluing some pair of local models from $\C_T (U_1)$ and $\C_T(U_2)$.
We will now argue that if $\C_T$ does not have surplus* structure, then this requirement cannot be satisfied; furthermore, we will only focus on the case of $\E_A$ since, as we discussed at the end of Section \ref{loc}, $\S_A$ only contains enough information to reconstruct the trivial monopole. 

We begin by considering how the isomorphism class (of the principal bundle $P$) of an arbitrary monopole can be reconstructed from $\C_A$ by compatibly gluing a pair of local models $A_1 \in \C_A (U_1)$ and $A_2 \in \C_A (U_2)$, i.e. by means of some gauge transformation $A_1 =A_2+g_{12}^{-1}dg_{12}$, where $g_{12}: U_{12} \simeq S^1\to U(1) \cong S^1$. 
In order to see how this works, it will be especially useful to recall from Section 4.1.1 that this isomorphism class corresponds to the monopole's topological charge $n \in \mathbb{Z}$, which can in turn be computed from its global force field $F$ by means of the formula $n = \frac{1}{2\pi} \int_{S^2} F$. Then, through a judicious application of Stoke's theorem, it is easy to see that $n = \frac{1}{2\pi} \int_{S^1} (A_1 - A_2) = \frac{1}{2\pi} \int_{S^1} g_{12}^{-1} dg_{12}$.
This is nothing other than the classical formula for the winding number (i.e. the number of times $g_{12}$ winds $S^1$ around $S^1$), and so we see that the topological charge of a monopole is determined by the `winding' (or homotopy class) of the particular choice of gauge transformation in $\overline{\C_A} (U_{12})$ that is used to glue $A_1$ and $A_2$.
It follows that in order to reconstruct monopole models with an arbitrary topological charge, our local theories need to contain all possible gauge transformations (and thus all possible homotopy classes), which is precisely the surplus* structure of $\C_A$ as compared to $\S_{[A]}$.

By contrast, if we take $\E_A$ as our local theory, then there is exactly one gauge transformation (and thus one possible choice of gluing) in $\overline{\E_A} (U_{12})$ between $A_1$ and $A_2$. 
Recall from the definition of $\E_A$ that this gauge transformation is described as $A_1 = A_2 + \lambda$, where $\lambda$ is a closed $1$-form.
It thus follows that the \textit{only} monopole models that one can construct by gluing local models are ones with the topological charge $n' = \frac{1}{2\pi} \int_{S^1} \lambda $ (where $n'$ is possibly non-zero if $\lambda$ is non-exact). 
To rephrase this crucial point: if the local theory does not have surplus* structure, i.e. the full set of gauge transformations, then all the other non-trivial monopole models (i.e.those with topological charge $n \neq n'$) cannot be constructed by gluing local models; this thus concludes our argument for the main claim: `(Richness) and (Locality) imply (Surplus*)'.

It is worth making two quick observations about this argument. First, we note that the argument holds up to categorical equivalence of the relevant local theories. In other words, the conclusion of the argument (it is not possible to simultaneously satisfy (Richness) and (Locality)) would have been the same if we had used $\S_F$ or $\S_{[A]}$ as our local theories instead of $\E_A$; and it would have been the same if we had used sk$\C_A$ instead of $\C_A$. 
Second, we note that one could of course advocate taking $\E_A$ as the relevant local theory if one were willing to give up either (Richness) or (Locality).
In other words, one could hold that the other monopole models are either physically irrelevant, or that they are non-local in the sense that they cannot be determined by the data of local theories. 
But by our lights, such a position is not only awkward but difficult to square with both physical practice and our intuitive understanding of fields as `local' objects.

Let us now summarise our explanation of the claim that `surplus* structure is an essential feature of gauge theory'.
The apparent tension in this statement stems from the idea that the surplus* structure is redundant structure and thus cannot be essential.
Our resolution of this tension turns on our above distinction between two different notions of `theory'. 
With respect to the first and limited notion, which we called a `local theory', surplus* structure is indeed redundant.
However, with respect to the second richer and more realistic notion, which we called an `expanded theory', such structure is necessary in order for the theory to perform the representational function jointly required by (Richness) and (Locality).
Thus, the tension only arises if one mistakenly tries to shoehorn gauge theory into the mould of the first narrow notion.

\centering
\section{Conclusion}
\label{conc}
\justify

Recall the motivation for the project of using category theory to represent scientific theories: a theory's content is not exhausted by its collection of models; it includes information about the relationships \textit{between} models. 
Category theory is a fruitful tool because it provides a natural framework in which to explicitly represent this information.
For the most part, the philosophical literature has only implemented this idea by treating a theory as a category, i.e. its models are represented as objects, and its inter-model relationships are encoded in morphims between objects.
For instance, in the categories of Section 3, the relationship of representational equivalence between models (gauge fields) is encoded in the isomorphisms (gauge transformations) of the category.

The argument of our paper can be seen as stemming from two related observations.
First, as we explained in Section 3, one important notion of `surplus' in gauge theory applies to theories in which gauge fields (corresponding to the same $F$) are taken to be representationally equivalent, and the gauge transformations are themselves included as part of the putative redundancies. We also showed that one can give a categorical analysis of this notion, which we called `surplus*'.  
Second, representational equivalence is only one kind of inter-model relationship; in order to capture other relationships such as global-local relationships, it is natural to move to a `theories as functors' framework, within which the `theories as categories' framework can be embedded.
In the context of gauge theory, such a move can be further motivated by the fact that field theories are typically taken to possess `locality', a property that can be precisely represented within the `theories as functors' framework.
As we argued above, these two observations are related by the conventional wisdom that there is a tight conceptual link between the `surplus' of gauge theory and representing the `locality' of fields. 
Thus, an immediate pay-off of the `theories as functors' framework is that it allows us to relate gauge theories on various spacetimes in order to address the following questions: (1) What is the correct analysis of `surplus' for the `theories as categories' of Section 3?, and (2) What is the representational role of this `surplus' in gauge theory more generally?
In Section 4, we used this relation to provide an independent check on our answer to (1), namely (Surplus*), thus confirming that our analysis really is responsive to desiderata stemming from physical practice, namely the need to simultaneously satisfy (Richness) and (Locality).

From a broader perspective that concerns the representational content of scientific theories, our argument here provides more evidence for the benefit of thinking about scientific theories in category theoretic terms. As we have seen, although it is the objects of the categories (the models) which directly represent possible systems, the morphisms can also feature in the representational content of $U(1)$ gauge theory. 
In particular: they play a role in representing how subsystems represented by objects of the theories-as-categories of Section 3 can be composed in the theories-as-functors framework. 
Consider, for example, the comparison between $\C_A$ and $\E_A$. The objects in these categories are the same, so it is only the morphisms, in particular the non-trivial automorphisms (and morphisms generated by these), which distinguish the two categories. 
And yet as our argument in the previous section showed, because of $\C_A$'s ability---and $\E_A$'s inability---to adequately represent how the local subsystems compose to form global systems (in a way that satisfies (Richness) and (Locality)), the theories have inequivalent representational content. 

To conclude, we wish to point out a striking parallel between our discussion of how scientific theories should be represented---which involves classifying models/theories up to some standard of sameness---and a far-reaching approach to classification that has been developed within pure mathematics (especially algebraic geometry). 
The central concept here is that of a `moduli space' $\mathcal{M}$ of objects (cf. our local models), i.e. a space whose points are in 1-1 correspondence with isomorphism classes of such objects (thus classifying them in a naive sense), and which is in addition rich enough to encode the structure of \textit{families} of such objects over a base space $B$ (cf. our families of local models parameterised by a contractible cover of $M$) in the sense that such families will be in 1-1-correspondence with maps from $B$ to $\mathcal{M}$.\footnote{For a history of the development of the `moduli space' concept, see \citep{Ji2016}, and for an elementary introduction that stresses the notion of `classification', see \citep{BenZvi}.}
Despite the fruitfulness of this notion, it was recognised very early on that such moduli spaces may not exist if the objects possess non-trivial automorphisms (cf. the non-trivial automorphisms of objects in $\C_A$), because the automorphisms will quite generically lead to various non-trivial families of objects over $B$, which will not be classified by maps from $B$ into a candidate $\mathcal{M}$.
Two standard strategies for addressing this problem are: (i) `rigidification', i.e. simply omitting the non-trivial automorphisms from the objects; and (ii) defining a richer structure called a `moduli stack', in which one does not merely assign (categorical) sets of objects to regions of the base space, but rather a groupoid that includes information about the non-trivial automorphisms of the objects, and thus allows one to keep track of the various non-trivial ways in which they `glue' into families of objects.

It will be clear from this heuristic discussion that taking $\E_A$ as one's local theory constitutes a rigidification of the models of gauge theory, whereas taking $\C_A$ as one's local theory is a key step in passing to the richer framework of moduli stacks---thus, from a broader mathematical perspective, one way of phrasing the moral of our philosophical discussion is that, in gauge theory, there are very strong grounds for thinking of a collection of models as a `stack'.\footnote{For a detailed presentation of how to describe a stack of gauge fields, see \citep{Benini2017}.}
For all that, there is still a disanalogy between these pure mathematical strategies and our physically-motivated discussion. 
In the mathematical case, to rigidify the objects is merely to simplify the problem so that one obtains a well-behaved moduli space.
By contrast, in our case, to rigidify the models is to say that one does not care either about representing (Richness), or about representing (Locality).

\centering
\section*{Appendix}
\label{appendix}
\justify
\addcontentsline{toc}{section}{\nameref{appendix}}
\G*

\begin{proof}
To see that $G$ is faithful notice that for all $A$ in $\S_A$, hom$(A,A)=\{1_A\}$, so the map $G_{A \to A}:$ hom$(A, A) \to$ hom$(G(A), G(A))$ has to be injective. For all $A, A'$ in $\S_A$ such that $A \neq A'$, hom$(A, A') = \varnothing$, and thus $G_{A \to A'}$ is also injective. So $G$ is faithful. 
To see that $G$ is essentially surjective notice that every $[A]$ in $\S_{[A]}$ is simply an equivalence class of gauge fields, and as such contains at least one $A$ such that $A$ is in $\S_A$ and $G(A)=[A]$.
To see that $G$ fails to be full consider an equivalence class of gauge fields $[A]$ in $\S_{[A]}$ that contains distinct gauge fields $A, A'$ (which are in $\S_A$). Since $\S_A$ contains only the identity morphisms, hom$(A,A') = \varnothing$. 
But, $G(A)=G(A')=[A]$, so hom$(G(A), G(A')) = \{1_{[A]}\}$. So the map $G_{A \to A'}$ is not surjective.\footnote{Notice that, since $\S_F \cong \S_{[A]}$, $G$ also shows that $\S_A$ has more structure* than $\S_F$. The proof of Proposition 3.2.1 follows the proof of Proposition 1 in \citep[p. 1044-1045]{W2}. This should make clear how our $\S_{[A]} \cong \S_F$ and $\S_A$ are analogues of his \textbf{EM$_1$} and \textbf{EM$_2$} respectively.} 
\end{proof}

\K*

\begin{proof}
$K$ is faithful for an analogous reason to $G$: for every object $[A]$ in $\S_{[A]}$, hom$([A], [A]) = \{1_{[A]}\}$ and every pair of distinct objects $[A]$ and $[A]'$, hom$([A], [A]') = \varnothing$. Thus both $K_{[A] \to [A]}$ and $K_{[A] \to [A]'}$ have to be injective.
To see that $K$ is essentially surjective notice that every $A$ in $\C_A$ is in some equivalence class $[A]$ in $\S_{[A]}$ and is isomorphic to every $A'$ which is also in that equivalence class. So either $K([A]) = A$, or $K([A]) = A'$ such that $A \cong A'$ in $\C_A$. This suffices to show that $K$ is essentially surjective.
To see that $K$ is not full consider an arbitrary $[A]$ in $\mathcal{S}_{[A]}$. By construction hom$([A], [A]) = \{1_{[A]}\}$. But, by the definition of the morphisms in $\C_{A}$ (given by definition (\ref{u1trans})), hom$(K([A]),K([A]))$ contains automorphisms which are not $1_{K([A])}$. \sloppy So the map $K_{[A] \rightarrow [A]}$ is not surjective.
\fussy
\end{proof}

\T*

\begin{proof}
$\tau$ is essentially surjective for the same reason as $G$.
To see that $\tau$ is full we consider an arbitrary pair of objects $A, A'$ in $\C_A$. If $A \ncong A'$ in $\C_A$ then $\tau(A) \neq \tau(A')$, in which case $\tau_{A \to A'}: \varnothing \to \varnothing$. If either $A=A'$, or $A \cong A'$ in $\C_A$, then $\tau(A) = \tau (A')$, so hom$(\tau(A), \tau(A')) = \{1_{\tau(A)}\}$. So in each case the map $\tau_{A \to A'}$ is surjective. 
To see $\tau$ is not faithful consider an arbitrary object $A$ in $\C_{A}$. By definition (\ref{u1trans}), hom$(A, A)$ contains automorphisms which are not $1_A$, and by the definition of $\mathcal{S}_{[A]}$, hom$(\tau (A), \tau (A)) = \{1_{\tau (A)}\}$. So the map $\tau _{A \rightarrow A}$ is not injective. 
\end{proof}

\M*

\begin{proof}
To see that $M$ is full and faithful note that both the domain and codomain categories contain only identity morphisms. If $[A] \neq [A]' $ then $M_{[A] \to [A]'}: \varnothing \to \varnothing$ and if $[A] = [A]'$ then $M_{[A] \to [A]}: \{1_{[A]}\} \to \{1_{M([A])}\}$; therefore, for every $[A], [A]'$ in $S_{[A]}$, $M_{[A] \to [A]'}$ is a bijection. To see $M$ is not essentially surjective note that by the definition of $M$, for each $[A]$ in $\mathcal{S}_{[A]}$, $M([A])=$ some $A$ in $\S_{A}$ such that $A \in [A]$. 
By construction there is also an $A' \neq A$ in $\S_{A}$ such that $A' \in [A]$, but hom$(A, A') = \varnothing$. 
Furthermore, if one takes a distinct equivalence class $[A^*]$ in $\S_{[A]}$ where $A'\notin [A^*]$, then $M([A^*])\ncong A'$ in $\S_A$. 
So $A'$ is not isomorphic to anything in the image of $M$.\\

\end{proof}

\centering

\section*{Acknowledgements}
\justify

Authors are listed in alphabetical order, the paper is fully collaborative. We are grateful to two anonymous referees for useful comments. We would also like to thank Jeremy Butterfield, Bryan Roberts, Alexander Schenkel, Christopher Schommer-Pries, and Stephan Stolz for helpful discussions and comments on previous drafts. Thanks also to audiences at the following workshops and conferences: \textit{The Physics Old and New: Historical and Critical Perspectives on Physics and the Philosophy of Nature} at the University of Notre Dame, \textit{Symmetries and Asymmetries in Physics} at Leibniz University in Hannover, and the \textit{Ninth European Congress of Analytic Philosophy} in Munich.

\begin{flushright}
\textit{
James Nguyen\\
History and Philosophy of Science\\ University of Notre Dame\\
Notre Dame, USA\\}
and\\
\textit{Centre for Philosophy of Natural and Social Science\\
London School of Economics and Political Science\\
London, UK\\
jnguyen4@nd.edu\\
\vspace{5mm}
Nicholas J. Teh\\
Department of Philosophy\\
University of Notre Dame\\
Notre Dame, USA\\
nteh@nd.edu\\
\vspace{5mm}
Laura Wells\\
Department of Mathematics\\
University of Notre Dame\\
Notre Dame, USA\\
lwells@nd.edu
}
\end{flushright}

\bibliographystyle{apalike}
\bibliography{archive}

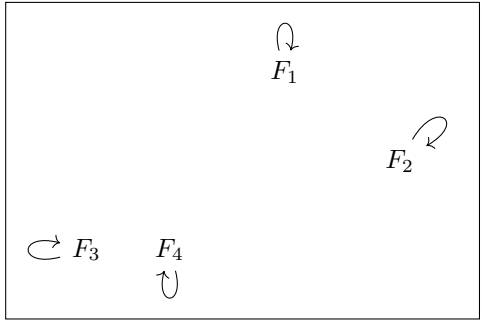
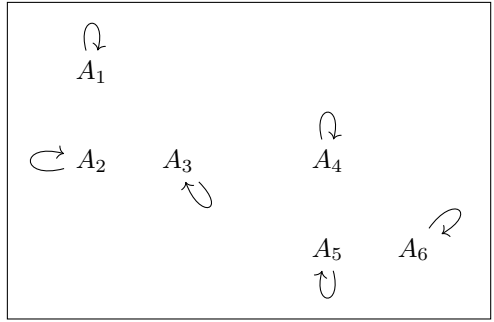
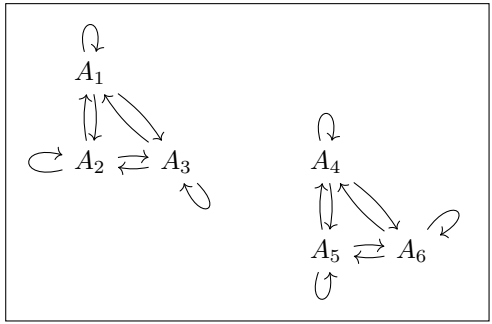
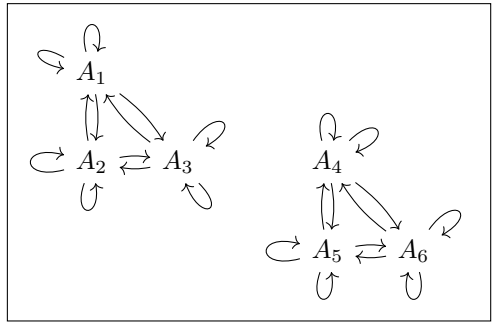
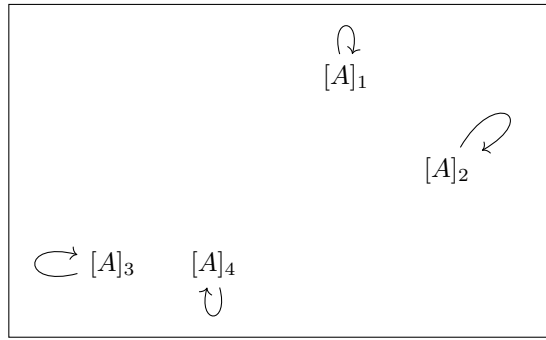
\begin{figure}[H]

    \centering
  
  \begin{subfigure}[b]{0.49\textwidth}
 
      \begin{tikzcd}[column sep=small, framed]
& &&F_1 \arrow[loop above] \\
&&&&& F_2 \arrow[in=30,out=60,loop] \\
F_3 \arrow[loop left] & F_4 \arrow[loop below]
\end{tikzcd}
       \caption{$\mathcal{S}_F$}
    \end{subfigure}
    \hfill
        \begin{subfigure}[b]{0.49\textwidth}
  \begin{tikzcd}[column sep=small,framed]
A_1 \arrow[loop above] \\
A_2  \arrow[loop left] & A_3  \arrow[in=290, out=315, loop] &&& A_4 \arrow[loop above]\\
&&&& A_5 \arrow[loop below]& A_6 \arrow[in=30, out=55, loop] 
\end{tikzcd}
       \caption{$\S_{A}$}
    \end{subfigure}

    \begin{subfigure}[b]{0.49\textwidth}
  \begin{tikzcd}[column sep=small,framed]
A_1  \arrow[loop above] \arrow[d, bend right=10, leftarrow] \arrow[d, bend left=10, rightarrow] \arrow[dr, bend right=10, leftarrow] \arrow[dr, bend left=10, rightarrow]\\
A_2 \arrow[r, bend right=10, leftarrow] \arrow[r, bend left=10, rightarrow] \arrow[loop left] & A_3 \arrow[in=290, out=315, loop]  &&& A_4 \arrow[loop above] \arrow[d, bend right=10, leftarrow] \arrow[d, bend left=10, rightarrow] \arrow[dr, bend right=10, leftarrow] \arrow[dr, bend left=10, rightarrow]\\
&&&& A_5 \arrow[r, bend right=10, leftarrow] \arrow[r, bend left=10, rightarrow]  \arrow[in=280,out=250,loop]& A_6 \arrow[in=30, out=55, loop] 
\end{tikzcd}
  \caption{$\mathcal{E}_{A}$}
    \end{subfigure}
    \hfill
        \begin{subfigure}[b]{0.49\textwidth}
      \begin{tikzcd}[column sep=small,framed]
A_1 \arrow[loop above]
\arrow[in=170,out=150,loop] \arrow[d, leftarrow, bend right=10] \arrow[d, rightarrow, bend left=10]\arrow[dr, leftarrow, bend right=10] \arrow[dr, rightarrow, bend left=10]\\
A_2 \arrow[r, leftarrow, bend right=10] \arrow[r, rightarrow, bend left=10] \arrow[loop left] \arrow[in=280,out=250,loop] & A_3 \arrow[in=30, out=55, loop] \arrow[in=290, out=315, loop] &&& A_4 \arrow[loop above] \arrow[in=20,out=45,loop] \arrow[d, leftarrow, bend right=10] \arrow[d, rightarrow, bend left=10] \arrow[dr, leftarrow, bend right=10] \arrow[dr, rightarrow, bend left=10]\\
&&&& A_5 \arrow[r, leftarrow, bend right=10]\arrow[r, rightarrow, bend left=10]  \arrow[loop left] \arrow[in=280,out=250,loop] & A_6 \arrow[in=30, out=55, loop] \arrow[loop below]
\end{tikzcd}
  \caption{$\C_{A}$}
    \end{subfigure}

       \begin{subfigure}[b]{0.49\textwidth}
      \begin{tikzcd}[column sep=small, framed]
& &&\lbrack A \rbrack_1 \arrow[loop above] \\
&&&& \lbrack A \rbrack_2 \arrow[in=30,out=60,loop] \\
\lbrack A \rbrack_3 \arrow[loop left] & \lbrack A \rbrack_4 \arrow[loop below]
\end{tikzcd}
       \caption{$\mathcal{S}_{[A]}$}
    \end{subfigure}
   
    \caption{Categories of Gauge Theories}\label{fig:cats}
\end{figure}

\begin{figure}[H]
\centering

\begin{tikzcd}[column sep=huge, row sep=huge]
& \mathcal{S}_F  \arrow[d, bend left=20, "H^{-1}"] & & \\
\C_{A}  \arrow[r, bend left=10, "\tau"] 
& \mathcal{S}_{[A]} \arrow[d, phantom, "\simeq" description]  \arrow[u, phantom, "\cong" description]  \arrow[d, bend left=20
, "I"] \arrow[u, bend left=20, "H"] \arrow[r, bend left=10, "M"] \arrow[l, bend left=10, "K"]
& \S_{A} \arrow[l, bend left=10, "G"]   \\
& \mathcal{E}_{A}  \arrow[u, bend left=20, "J"] &
\end{tikzcd}
\caption{Functorial Relationships\label{fig:rels}}

\end{figure}
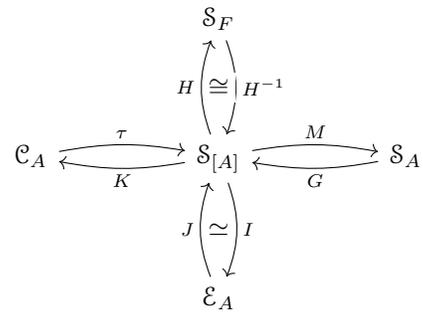

\end{document}